%% file: main_jumbleequivalence.tex
\documentclass[envcountsame, cleveref]{lipics-v2021}
\usepackage{complexity}
\usepackage{amsmath,amsthm}

\hideLIPIcs
\nolinenumbers

\newtheorem{problem}[]{Problem}

\usepackage{amssymb}
\usepackage{multirow}
\usepackage{tikz-network}
\usepackage{tikz}
\usetikzlibrary{arrows,automata,positioning,calc}
\usepackage{comment}
\usepackage{algorithm}
\usepackage[noend]{algpseudocode}

\usepackage{hyperref}

\usepackage{pdflscape}

\input{shortcuts}

\newcommand{\orcidauthorA}{0000-0002-1531-7970} 
\newcommand{\orcidauthorB}{0009-0005-1145-5806} 
\newcommand{\orcidauthorC}{0000-0003-3153-0314} 
\newcommand{\orcidauthorD}{0000-0003-1008-4642} 
\newcommand{\orcidauthorE}{0009-0000-3906-2581} 

\title{Jumbled Simon's Equivalence and Embeddings}
 
\titlerunning{}

\author{Pamela Fleischmann}
{Kiel University, Department of Computer Science, Germany}
{fpa@informatik.uni-kiel.de}
{\orcidauthorA}{}

\author{Annika Huch}
{Kiel University, Department of Computer Science, Germany}
{ahu@informatik.uni-kiel.de}
{\orcidauthorB}{}

\author{Sungmin Kim}
{Department of Computer Science, Yonsei University, Republic of Korea}
{rena\_rio@yonsei.ac.kr}
{\orcidauthorC}{}

\author{Tore Koß}
{University of Liege, Institute of Mathematics, Belgium}
{tore.koss@uliege.be}
{\orcidauthorD}{this research was supported by the University of Liege under Special Funds for
Research, IPD-STEMA Programme}   

\author{Emil Wiards}
{Kiel University, Department of Computer Science, Germany}
{stu249700@mail.uni-kiel.de}
{\orcidauthorE}{}

\authorrunning{P. Fleischmann \etal}

\EventEditors{John Q. Open and Joan R. Acces}
\EventNoEds{2}
\EventLongTitle{42nd Conference on Very Important Topics (CVIT 2016)}
\EventShortTitle{CVIT 2016}
\EventAcronym{CVIT}
\EventYear{2016}
\EventDate{December 24--27, 2016}
\EventLocation{Little Whinging, United Kingdom}
\EventLogo{}
\SeriesVolume{42}
\ArticleNo{23}



\Copyright{Pamela Fleischmann, Annika Huch, Sungmin Kim, Tore Koß, Emil Wiards} 

\begin{document}

\maketitle

\input{abstract}

\addtocounter{page}{-1}
\newpage
\section{Introduction}
\input{intro}
\section{Preliminaries}
\input{prelims}

\section{Testing Jumbled Simon's Equivalence}
\input{automaton}

\section{The Jumbled Binomial Coefficient}

\input{binom.tex}

\section{Conclusion}
\input{conclusion.tex}

\newpage

\bibliography{bib.bib}

\end{document}

%% file: shortcuts.tex
\def\ta{\mathtt{a}}
\def\tb{\mathtt{b}}
\def\tc{\mathtt{c}}

\def\etal{\emph{et al. }}

\def\nth#1{#1$^{\text{th}}$}

\def\N{\mathbbm{N}}

\def\abs#1{|#1|}
\def\len#1{|#1|}

\renewcommand{\epsilon}{\varepsilon}
\newcommand{\eps}{\varepsilon}
\renewcommand{\phi}{\varphi}

\def\N{\mathbb{N}} 

\DeclareMathOperator{\al}{alph}

\usepackage{bbm}

\newcommand{\ScatFact}{\mathbb{S}}
\newcommand{\JumbledScatFact}{\mathbb{J}}
\newcommand{\JScatFact}{\JumbledScatFact}
\newcommand{\JSF}{\JumbledScatFact}
\DeclareMathOperator{\SJSF}{\overline{\JScatFact}}

\newcommand{\CSF}{\operatorname{CSF}}
\newcommand{\lcsf}{\operatorname{lcsf}}
\newcommand{\LCSF}{\operatorname{LCSF}}

\newcommand{\LCS}{\LCSF}

\newcommand{\Emb}{\mathbb{E}}

\DeclareMathOperator{\letters}{alph}

\newcommand{\jsimkl}{\overset{J}{\sim}_{k,\ell}}
\newcommand{\jsim}{\overset{J}{\sim}}
\newcommand{\njsim}{\overset{J}{\nsim}}

\usepackage{shuffle}

\def\shuffle{\sqcup\mathchoice{\mkern-7mu}{\mkern-7mu}{\mkern-3.2mu}{\mkern-3.8mu}\sqcup}

\newcommand\restr[2]{{
  \left.\kern-\nulldelimiterspace 
  #1 
  \vphantom{\big|} 
  \right|_{#2} 
  }}

\newcommand{\newterm}[1]{\textnormal{#1}}

\usepackage{todonotes}

\newcommand{\todotk}[1]{\todo[color=green!20,inline]{Tore: #1}}

\newcommand{\todoski}[1]{\todo[color=pink!20,inline]{Sungmin: {#1}}}

\newcommand{\vect}[1]{\overrightarrow{#1}}
\newcommand{\XYranker}[4]{{#1}(#2,#3,#4)}

\newcommand{\sharpp}{\textsf{\upshape \#P}}
\newcommand{\adj}[1]{\texttt{\upshape Adj}({#1})}
\newcommand{\emb}{\texttt{\upshape Emb}}
\newcommand{\cho}[1]{\texttt{\upshape Cho}(#1)}

\newcommand{\jbinom}[3][]{\genfrac{[}{]}{0pt}{}{#2}{#3}_{#1}}

\newcommand{\loosejumbledembeddings}{\Emb}
\newcommand{\jumbledembeddings}{\Emb}
\newcommand{\loosejscatfact}{\JScatFact}
\newcommand{\restrict}[2]{#1|_{#2}}

\newcommand{\embeddings}{\operatorname{E}}
\newcommand{\embeddingword}{w}

%% file: abstract.tex
\begin{abstract}
    \emph{Jumbled scattered factors} are a generalisation of scattered factors (or subsequences) where a jumbled scattered factor is obtained from another word by deleting and rearranging letters. For limiting the number of rearrangements, we distinguish $\ell$-jumbled scattered factors where $\ell$ limits the number of letters that may occur rearranged. For instance, $\mathtt{car}$, $\mathtt{start}$, $\mathtt{rat}, \mathtt{cat}$ are jumbled scattered factors of $\mathtt{abstract}$ even though $\mathtt{car}$, $\mathtt{start}$ and $\mathtt{cat}$ are no scattered factors and each need one jumble. An example for a jumbled scattered factor with $2$ jumbles is $\mathtt{tatar}$. In this work, we investigate two problems on jumbled scattered factors from an algorithmic perspective. First, we give algorithms for testing jumbled Simon's equivalence where two words are equivalent if they share the same $\ell$-jumbled scattered factors up to length $k$. Second, we investigate how to count the possible occurrences of a jumbled scattered factor -- the jumbled binomial coefficient. After defining the jumbled binomial coefficient, we study jumbled embeddings and factorisations of such as well.
\end{abstract}

%% file: intro.tex
A reoccurring topic in computer science, data science and bioinformatics is the analysis of data. One theoretical model for data are \emph{scattered factors} (also called subsequences, subwords) which are obtained by deleting some letters in a given string while preserving the order of the remaining letters. Formally, a \emph{scattered factor} of a word $w$ is
another word $u$ of length $m$ such that there exist (possibly empty) words $v_1, \dots, v_{m+1}$ satisfying $w = v_1 u[1] v_2 \cdots v_m u[m] v_{m + 1}$. 
For example, $\mathtt{cat}$, $\mathtt{care}$, and $\mathtt{sad}$ are scattered factors of
$\mathtt{scattered}$, while $\mathtt{dare}$ and $\mathtt{dart}$ are not.
Scattered factors are of ongoing research interest, dating back to the early investigations by Higman \cite{higman1952ordering} and Haines \cite{haines1969free} in the 1950s and 1960s.
In 1975, Imre Simon continued the research on scattered factors (\emph{piecewise subwords}) in the context of piecewise-testable languages \cite{simon1972hierarchies,DBLP:conf/automata/Simon75}. Since then, scattered factors attracted widespread attention in many fields such as combinatorics on words \cite{DBLP:journals/fuin/KoscheKMS22,DBLP:journals/jcss/MateescuSY04}, algorithms and stringology \cite{DBLP:journals/tcs/Baeza-Yates91,DBLP:journals/algorithmica/BannaiIKKP24,DBLP:conf/soda/BringmannK18,DBLP:journals/aam/ForresterM23,DBLP:journals/jacm/Maier78,DBLP:journals/jacm/WagnerF74}, logics \cite{DBLP:journals/lmcs/BaumannGTZ23,DBLP:conf/lics/HalfonSZ17,10.1007/978-3-030-50026-9_21,10.1007/978-3-030-17127-8_20}, and language and automata theory \cite{DBLP:conf/isaac/AdamsonFHKMN23,DBLP:journals/ipl/KarandikarKS15,DBLP:journals/lmcs/KarandikarS19,simon1972hierarchies,DBLP:conf/automata/Simon75,zetzsche:LIPIcs.ICALP.2016.123}. Further applications of scattered factors can be found in bioinformatics \cite{DBLP:journals/bib/DoLL21,DBLP:conf/aaai/WangCG20}, formal
software verification \cite{DBLP:conf/lics/HalfonSZ17,zetzsche:LIPIcs.ICALP.2016.123}
and database theory \cite{artikis2017complex,FrochauxK23,Kleest-Meissner22,Kleest-Meissner23,SchmidSIGMOD}.
The fundamental work of Simon \cite{simon1972hierarchies} also introduced a congruence relation that is nowadays called \emph{Simon's congruence}. Two words are Simon $k$-congruent for a natural number $k$, if they share the same scattered factors up to length $k$. Recent developments consider problems of testing and deciding Simon's congruence \cite{DBLP:conf/mfcs/FleischerK18,DBLP:conf/stacs/GawrychowskiKKM21,DBLP:conf/dlt/KimHKS23}.

Ongoing research on scattered factors motivated the introduction of \emph{jumbled scattered factors}, a generalisation of scattered factors which considers not only
scattered factors but also \emph{rearrangements} (see e.g. \cite{DBLP:conf/cwords/Puzynina19,DBLP:journals/tcs/KarhumakiPRW17} for abelian equivalence of words) of them \cite{DBLP:conf/dlt/FleischmannHKK25}. A word $u$ of length $m$ is a jumbled scattered factor of
$w$ if $w = v_1 u[i_1] \cdots v_m u[i_m] v_{m + 1}$ for some permutation $i_1 i_2 \cdots i_m$ and (possibly empty) words $v_1, \dots, v_{m+1}$. For example,
apart from all scattered factors, $\mathtt{dare}$ and $\mathtt{dart}$ are jumbled scattered
factors of $\mathtt{scattered}$ but $\mathtt{retreat}$ is not since $\mathtt{r}$ occurs only once in $\mathtt{scattered}$. Further, they also considered restrictions on the
number of jumbles, i.e. $u$ is $\ell$-jumbled in $w$ if at most $\ell$
characters are out of order or, in other words, there is a common
scattered factor of $u$ and $w$ of length $|u| -\ell$. In that sense, the word $\mathtt{dare}$ is $1$-jumbled and $\mathtt{dart}$ is $2$-jumbled
in $\mathtt{scattered}$. 
In general, there is an inverse relation between the
required number of jumbles and the longest common scattered
factors. Simon's congruence and universality were also generalised to
jumbled scattered factors in the hope that insights on the jumbled
version translate to the non-jumbled case; though, notably, the
jumbled version of Simons congruence is only an equivalence and not a
congruence relation \cite{DBLP:conf/dlt/FleischmannHKK25}.


While thus-far-mentioned research mainly focused on the existence of
(jumbled) scattered factors, there has also been research on the
number of occurrences of a scattered factor in a word. Here, separate
occurrences are usually modelled using word embeddings,
i.e. strictly increasing maps from the set of positions of a scattered factor $u$ into
the set of positions of $w$  such that their letters match. Typically, $\binom{w}{u}$ denotes
the number of distinct embeddings of $u$ in $w$~\cite{lothaire}. For example,
we get $\binom{\mathtt{banana}}{\mathtt{ana}} = 4$. However, there is no such research
on jumbled scattered factors as of yet.



\textbf{Our Contribution.} In this work, we develop the work on jumbled scattered factors in two directions: on the one hand we investigate the problem of testing jumbled $k$-Simon-equivalence by presenting an algorithm for fixed $k$ and $\ell$ using finite automata and another one harnessing the connection between jumbled scattered fators and longest common scattered factors. On the other hand, 
we investigate the problem of counting the amount of occurrences of a given jumbled scattered factor in a given word.

%% file: prelims.tex
Let $\N$ denote the natural numbers starting with $1$ and let $\N_0 :=\N\cup\{0\}$.
Further, let $[n]:=\{1,\ldots,n\}$ for some $n\in\N$ and $[n]_0=[n]\cup\{0\}$.
An {\em alphabet} $\Sigma$ is a finite, non-empty set whose elements are called {\em letters}.
As $\Sigma$ is finite, it can be assumed ordered and w.l.o.g. we take $\Sigma:=\{\ta_1,\dots,\ta_{\sigma}\}$
for some $\sigma\in\N$ with $\ta_i<\ta_{i+1}$ for all $i\in[\sigma-1]$.
A finite concatenation of letters from $\Sigma$ is called a {\em word}. The set of words $\Sigma^{\ast}$ with concatenation as binary operation and the empty word $\varepsilon$, i.e., the word without any letters, as the neutral element, forms a free monoid. 
The number of letters in a word $w\in\Sigma^{\ast}$ is called its {\em length} and denoted by $|w|$; this implies $|\varepsilon|=0$. For all $k\in\N_0$, let $\Sigma^k :=\{w\in\Sigma^{\ast}|\,|w|=k\}$ and define $\Sigma^{\leq k}:=\{w\in\Sigma^{\ast}|\,|w|\le k\}$.
The letter of $w\in\Sigma^{\ast}$ at position $i\in[|w|]$ is denoted by $w[i]$. For $w\in\Sigma^{\ast}$ and $\ta\in\Sigma$ define
$|w|_{\ta}:=|\{i\in[|w|]\mid w[i]=\ta\}|$ to be the number of occurrences of $\ta$ in $w$.
The {\em Parikh vector}~$p(w)$
of $w\in\Sigma^{\ast}$
is defined as
$(|w|_{\ta_1},\ldots,|w|_{\ta_{\sigma}})\in\N_0^{\sigma}$.
We denote the~\nth{$i$} entry of $p(w)$ by $p(w)[i]$,
for each $i \in [\sigma]$. Note that two words have the same Parikh vector if and only if one word can be obtained
by rearranging letters of the other.
For example, $\mathtt{east}$ and $\mathtt{seat}$
have the same Parikh vector.
Given two words $u,v\in\Sigma^{\ast}$, we say $p(u)\leq p(v)$ (or $p(u) < p(v)$) if $|u|_{\ta_i}\leq|v|_{\ta_i}$ (or respectively $|u|_{\ta_i}<|v|_{\ta_i}$) for all $i\in[\sigma]$. 

For any set $A$,
let us denote the mapping which maps each element of $A$ to itself as~$\texttt{id}_A:A\to A$. For sets $A,B,C$ such that $A\subset B$ and a mapping $f:B\to C$ we denote by $f|_A:A\to C$ the \emph{restriction of $f$ to $A$}, i.e. the mapping which satisfies $f|_A(a)=f(a)$ for all $a\in A$. Further, we denote the composition of mappings by $\circ$, i.e. for three sets $A,B,C$ and mappings $g:A\to B$ and $f:B\to C$ we denote by $f\circ g:A\to C$ the mapping which satisfies $(f\circ g)(a) = f(g(a))$ for all $a\in A$. Moreover, we call any mapping $f:A\to B$ \emph{bijective} (\emph{surjective}; \emph{injective}) if for any $b\in B$ there is exactly (at least; at most) one $a\in A$ such that $f(a)=b$ and recall that $f$ is bijective if and only if there is an inverse mapping $f^{-1}:B\to A$ satisfying $f^{-1}\circ f = \texttt{id}_A$ and $f\circ f^{-1} = \texttt{id}_B$.

For the algorithmic results we use the standard 
computational model RAM with logarithmic word-size (see, e.g. \cite{KarkkainenSB06}), i.e. we 
follow a standard assumption from stringology, that if $w$ is the 
input word for our algorithms, then we assume $\Sigma=\letters(w)=\{ 1,2, \ldots, \sigma\}$.

We define {\em scattered factors}, also known as subsequences or subwords, which will built the fundament for their generalised version, the jumbled scattered factors.

\begin{definition}\label{def:SF}
  A word $u \in \Sigma^*$ is called a \newterm{scattered factor}
  of $w\in\Sigma^{\ast}$ if 
  we can delete letters from $w$
  and obtain $u$.
  Let $\ScatFact(w)$ be the set
  of all scattered factors of $w$,
  and let $\ScatFact_k(w)$ and $\ScatFact_{\leq k}(w)$
  be the set of scattered factors respectively
  with length exactly $k$ and length at most $k$.
\end{definition}

\begin{definition}
	Let $v, w \in \Sigma^*$.
    The set of {\em common scattered factors} of $v$ and $w$
    is defined as
	$\CSF(v, w) := \ScatFact(v) \cap  \ScatFact(w)$.
    Further, let us denote the set of {\em longest common scattered factors of $v$ and $w$}
    with $\LCSF(v,w)$.
    Their length is denoted by $\lcsf(v, w)$.
\end{definition}

For instance, \texttt{abca} is a scattered factor of \texttt{cbacbacba} but \texttt{bccc} is not.
Furthermore,
for $v=\texttt{aabbcc}$ and $w=\texttt{ccbbaa}$,
we have $\CSF(v,w)=\{\varepsilon,
\texttt{a},\texttt{b},\texttt{c},
\texttt{aa},\texttt{bb},\texttt{cc}\}$.
Note that the letters of the word~\texttt{abc}
occur in $w$ but not in the same order.
Thus, the concept of \emph{jumbled scattered factors}
was introduced to describe words whose letters appear
out-of-order in a larger word.

\begin{definition}\label{jscatfact}
  Let $w \in \Sigma^*$. A word $u \in \Sigma^*$ is called a \newterm{jumbled scattered factor} of $w$ if there exists $v \in \ScatFact(w)$ such that $p(u) = p(v)$. The set of jumbled scattered factors of $w$ is denoted by $\JScatFact(w)$. 
\end{definition}

For example, any word
that consists of at most two occurrences of
of each \texttt{a},
\texttt{b}, and \texttt{c}
is a jumbled scattered factor of \texttt{ccbbaa}.
However, this definition is extreme
in the sense that whether or not
a word~$u$ is a jumbled scattered factor
of another word~$w$
depends only on their Parikh vectors,
which means that the order of letters in $u$ nor $w$
do not matter.
Indeed, $u$ is a jumbled scattered factor
of $w$
if and only if $p(u)\le p(w)$.
Therefore, we formalise the extent
to which letters in $u$ appear
out-of-order in $w$
as the number of letters to jumble in $u$
to obtain a scattered factor of $w$.

\begin{definition}\label{def:jidx}
  Let $w \in \Sigma^*$. The word $u \in \Sigma^*$ is called a \newterm{scattered factor
  of $w$ with $\ell \in \N_0$ jumbles} (or $u$ is $\ell$-jumbled in $w$) 
  if $p(u)\leq p(w)$ and $\lcsf(u,w)\geq \abs{u} - \ell$.
  The set of scattered factors of $w$ needing at most $\ell$ jumbles is denoted by $\JScatFact(w,\ell)$.
  Again, the restriction to jumbled scattered factors of length exactly (at most) $k$ is denoted by a subscript $\JSF_k(w,\ell)$
  $(\JSF_{\leq k}(w,\ell))$.
\end{definition}

In \cite{DBLP:conf/dlt/FleischmannHKK25}, a more restrictive variant of jumbled scattered factors than in \cref{def:jidx} is considered: there is an additional requirement $\ell \in [|u|-1]_0$ with the intention that every time at least one letter of the jumbled scattered factor fits and one letter can not be ``out of order''. This condition implies the anomaly that $\varepsilon \in \ScatFact(w)$ but $\varepsilon\notin\JScatFact(w,0)$ which does not match the intuition that the set of words obtained by not jumbling any letter is equal to the set of classical scattered factors. Since we are interested in a natural extension of the classical scattered factors, especially in the case of jumbled embeddings, we drop the condition $\ell \in [|u|-1]_0$ which results in $\JSF(w,0) = \ScatFact(w)$ for all $w\in \Sigma^*$. Further, if $u\neq\varepsilon$ and $p(u)\leq p(w)$ then $\lcsf(u,w)\geq 1$, thus, $u$ is $(\len u-1)$-jumbled in $w$ if and only if it is $\ell$-jumbled in $w$ for any $\ell\geq \len u$.
Accordingly, we adjust the definition of
$\ell$-jumbled scattered factors of a word~$w$,
where $\ell$ is the degree
a jumbled scattered factor deviates from being
a scattered factor of $w$.

Note that, for $w\in\Sigma^\ast$, $u \in \JSF(w, \ell)$ implies that every scattered factor of $u$ is at most $\ell$-jumbled in $w$. By \Cref{def:jidx} we have that \texttt{owl} is $1$-jumbled and \texttt{wolf} is not only $2$-jumbled but also $3$-jumbled depending on which \texttt{l} in \texttt{cauliflower} is chosen. For a more precise description of the different possibilities in which a jumbled scattered factor appears in a word, we extend the notion of embeddings introduced in \cite{DBLP:conf/isaac/DayKMS22}.

\begin{definition}
    Let $w \in \Sigma^*$, $\ell \in \N_0$ and $u\in\JSF(w,\ell)$. An \newterm{$\ell$-jumbled embedding} of $u$ into $w$ is an injective function $e : [|u|] \to [|w|]$ such that $u[i] = w[e(i)]$ for all $i \in [|u|]$ and there exists a subset $S \subseteq [|u|]$ of size $|S| \geq |u| - \ell$ such that the restriction $e|_S$ is strictly increasing.
    The set of $\ell$-jumbled embeddings of $u$ into $w$ is denoted by $\loosejumbledembeddings(u,w,\ell)$.
\end{definition}

Consider for example $\mathtt{wolf}$ as a $2$-jumbled scattered factor in $\mathtt{cauliflower}$. A $2$-jumbled embedding of $\mathtt{wolf}$ is, e.g. $\{1 \mapsto 9, 2 \mapsto 8, 3 \mapsto 4, 4 \mapsto 6\}$ with the strictly increasing restriction $\{3 \mapsto 4, 4 \mapsto 6\}$. 
Notice that, analogously to classical embeddings of scattered factors, jumbled embeddings may not be unique within a word, e.g. $\{1 \mapsto 9, 2 \mapsto 8, 3 \mapsto 7, 4 \mapsto 6\}$ is a $3$-jumbled embedding of $\mathtt{wolf}$ in $\mathtt{cauliflower}$.

Based on the longest common scattered factor we define the jumbled neighbourhood.

\begin{definition}
    For $w\in\Sigma^\ast, \ell \in \N_{0}$ the \emph{$\ell$-jumbled neighbourhood of $w$} is defined as the language $N(w,\ell) := \{v\in\Sigma^{|w|}\mid \lcsf(v,w) \geq |w|-\ell\}$.
\end{definition}

The following result links the jumbled scattered factors of a word with the downward closure of the $\ell$-jumble-neighbourhood. 
\begin{proposition}
    Let $w \in \Sigma^{*}$, $k, \ell \in \N_0$ with $\ell \leq k \leq |w| $. Then
    $\JScatFact_{\leq k}(w, \ell)=
    \bigcup_{x\in N(w,\ell)}\ScatFact_{\leq k}(x)$.
\end{proposition}
\begin{proof}
    We first prove that $\JScatFact_{\leq k}(w, \ell)\subseteq
    \bigcup_{x\in N(w,\ell)}\ScatFact_{\leq k}(x)$.
    Therefore, let $u\in \JScatFact_{\leq k}(w, \ell)$. 
    Further, there must exist a common scattered factor $v$ of $u$ and $w$
    with length at least $|u|-\ell$.
    Furthermore, since $p(v)\le p(u)\le p(w)$,
    we have $p(u)-p(v)\le p(w)-p(v)$
    where $-$ is the element-wise subtraction operation.
    Therefore,
    for all $\ta\in\Sigma$,
    we can reposition $p(u)[\ta]-p(v)[\ta]$ occurrences of $\ta$ in $w$
    so that $u$ is a scattered factor of the new word.
    The sum of all elements in $p(u)-p(v)$
    is at most $\ell$,
    thus $u\in\bigcup_{x\in N(w,\ell)}\ScatFact_{\leq k}(x)$.

    We now prove that $\JScatFact_{\leq k}(w, \ell)\supseteq
    \bigcup_{x\in N(w,\ell)}\ScatFact_{\leq k}(x)$.
    Let $u\in \bigcup_{x\in N(w,\ell)}\ScatFact_{\leq k}(x)$.
    Then, by definition,
    there exists a word ~$x\in N(w,\ell)$ such that $u$ is a scattered factor of $x$. Since $x\in N(w,\ell)$, we have $\lcsf(x,w)\geq |x|-\ell$. Let $\ell'\leq\ell$ such that $\lcsf(x,w) = |x|-\ell'$, fix a longest common scattered factor $x[i_1]x[i_2]\cdots x [i_{|x|-\ell'}]$ and let $\{j_1,j_2,\ldots,j_{\ell'}\} = [|x|]\setminus \{i_1,\ldots,i_{|x|-\ell'}\}$.
    For any scattered factor $u=x[i'_1]x[i'_2]\cdots x[i'_{|u|}]$ in $x$, the intersection $\{j_1,\ldots, j_{\ell'}\}\cap\{i'_1,\ldots,i'_{|u|}\}$ has at most $\ell'$ elements. Further, the positions not contained in that intersection form a scattered factor of a longest common scattered factor of $x$ and $w$. Consequently, we obtain a common scattered factor of $u$ and $w$ by removing at most $\ell'\leq\ell$ many letters from $u$ and, thus, $u$ is $\ell$-jumbled in $w$.
    Let $v$ be the scattered factor of $u$
    obtained by removing $\ell$ letters from $u$
    including all the letters that
    correspond to the jumbled letters in $x$.
    Then, $v\in \ScatFact(u)$
    with $|v| \geq |u|-\ell$.
    Furthermore, $p(u)\le p(x)=p(w)$,
    and therefore $u\in \JScatFact_{\leq k}(w,\ell)$.
\end{proof}

\begin{lemma}
    For~$w\in \Sigma^n$
    and $\ell\in [n]_0$,
    there are at most
   $\binom{n}{\ell}^2\ell!$
    words in the $\ell$-jumble neighbourhood
    of $w$.
\end{lemma}
\begin{proof}
    We first choose from $w$ the
    $\ell$ letters to jumble.
    There are at most $\binom{n}{\ell}\ell!$ choices
    in choosing and ordering
    $\ell$ positions of $w$.
    Then, we sample with replacement $\ell$ positions
    in the remaining $n-\ell+1$ gaps in the word
    and insert the letters
    in the chosen places
    in the designated order.
    This gives us $\binom{n-\ell+1+\ell-1}{\ell}=\binom{n}{\ell}$ combinations
    for each of the orderings.
    Therefore, the statement follows.
\end{proof}


We end this section by recalling the jumbled variant of Simon's congruence \cite{DBLP:conf/automata/Simon75}.
This leads to the problem of investigating the jumbled variant $\jsimkl$ of Simon's congruence relation.

\begin{definition}
	Let $v,w \in \Sigma^*$ and $k,\ell \in \N$. We say that $v$ is \emph{$\ell$-jumble $k$-Simon equivalent} to $w$, for short $v \jsimkl w$ iff $\JScatFact_{\leq k}(v, \ell) = \JScatFact_{\leq k}(w, \ell)$. We denote the equivalence class of $w$ under $\jsimkl$ by $[w]_{k,l} = \{v\in\Sigma^\ast\mid v\jsimkl w\}$.
\end{definition}

\begin{problem}\label{problem:simon}
	Given two words $v,w \in \Sigma^*$ and two integers $k,\ell \in \N$ is $v \jsimkl w$?
\end{problem}

%% file: automaton.tex
Although $\jsimkl$ is an extension of Simon's congruence,
it is not a congruence relation.
For example,
consider the words~\texttt{abcd} and \texttt{dabc},
as well as the words~\texttt{efgh} and \texttt{hefg}.
Each pair of words is 1-jumble 4-Simon equivalent.
However, the words~\texttt{abcdefgh}
and \texttt{dabchefg} obtained
by concatenating, respectively the first and second words
in each pair, are not 1-jumble 4-Simon equivalent
due to the counterexample \texttt{cdgh}.
Nevertheless, a natural question is the one of  Problem~\ref{problem:simon}, i.e. testing the $\ell$-jumbled $k$-Simon's equivalence for two given words. 
First, we present combinatorial insights on $\jsimkl$ before considering procedures for testing the equivalence of two words.

\begin{proposition} \label{prop:jumbled-simons-equiv-max-ell-characterisation}
    Let $k \in \N$ and $w, w' \in \Sigma^*$. Then, we get $w \jsim_{k, k-1} w'$ if and only if either $p(w)[i] = p(w')[i]$ or $p(w)[i], p(w')[i] \ge k$ for all $i \in [\sigma]$.
\end{proposition}
\begin{proof}
    ``$\Rightarrow$'': We show this by contraposition. Let $k \in \N$ and $w, w' \in \Sigma^*$ such that there exists $i \in [\sigma]$ with w.l.o.g. $p(w)[i] < k$ and $p(w)[i] < p(w')[i]$. Let $u = \ta_i^n$ where $n = \min \{p(w')[i], k\}$. Then, $\len{u} \le k$ and $u \in \JSF_{\le k}(w', k-1)$ but $u \not\in \JSF_{\le k}(w, k-1)$.
    
    ``$\Leftarrow$'': We show this by contraposition. Let $k \in \N$ and $w, w' \in \Sigma^*$ with $w \not\jsim_{k,k-1} w'$. First, we look at the case that $\len{w} < k$ or $\len{w'} < k$. Let w.l.o.g. $\len{w} < k$. Then, there is $i \in [\sigma]$ such that $p(w)[i] \neq p(w')[i]$ since otherwise we would find $p(w) = p(w')$ which implies $w \jsim_{k,k-1} w'$. Next, let $\len{w}, \len{w'} \ge k$. Then, according to \cite[Lemma 11]{DBLP:conf/dlt/FleischmannHKK25}, we find w.l.o.g. $u \in \JSF_k(w,k-1)$ such that $u \not\in \JSF_k(w',k-1)$. Thus, $p(u) \le p(w)$ but $p(u) \not\le p(w')$. So there exists $i \in [\sigma]$ with $p(u)[i] > p(w')[i]$. We get $k = \len{u} \ge p(u)[i] > p(w')[i]$.
\end{proof}

\begin{corollary}
    Let $k \in \N$. We get the index $|\Sigma^* /\mathord{\jsim_{k,k-1}}| = (k + 1)^\sigma$ for the jumbled Simon's equivalence $\jsim_{k,k-1}$.
\end{corollary}
\begin{proof}
    From \Cref{prop:jumbled-simons-equiv-max-ell-characterisation}, we know that the equivalence classes of $\jsim_{k,k-1}$ are characterised by $\min \{p(w), k\}$ for any $w \in \Sigma^*$. Thus, we have that $|\Sigma^* / \mathord{\jsim_{k,k-1}}| = |[k]_0^\sigma| = (k+1)^\sigma$.
\end{proof}

Previously, Fleischmann et al.~\cite{DBLP:conf/dlt/FleischmannHKK25}
discovered equivalence classes of size $1, |\Pi(w)|$, and $\infty$ where $\Pi(w)$ denotes the set of words with the same Parikh vector as some representative $w$.
We prove that there are other equivalence classes
whose cardinality does not belong in one of the three cases
by the following example.

\begin{example}
    The equivalence class $[\mathtt{acbab}]_{3,1}$ is an example for a class that has finite size which is not the number of permutations of $\mathtt{acbab}$. 
    One may verify that 
    \begin{align*}
    [\mathtt{acbab}]_{3,1} =& (\mathtt{abab} \shuffle \tc) \cup (\mathtt{abba} \shuffle \tc) \cup (\mathtt{baab} \shuffle \tc) \cup (\mathtt{baba} \shuffle \tc) \\ &\cup \{\mathtt{acabb}, \mathtt{aabcb}, \mathtt{bcbaa}, \mathtt{bbaca}\}.
    \end{align*}

    First, we justify why all of those words have the same $1$-jumbled scattered factors of length $3$. Every word produced by a shuffle is $2$-universal w.r.t. $\ta$ and $\tb$, i.e. they contain every word of length $3$ that has a component-wise smaller Parikh vector than $(2,2,1)$. The same holds for the four words in the second line due to the placement of $\tc$. 
    Second, we justify that there are no other words contained in this set. Any word with more occurrences of  $\ta, \tb$ or $\tc$ does contain $\ta^3, \tb^3$, or $\tc^2$ as a scattered factor. Further, any other position of $\tc$ than the four within the second row would lead to fewer jumbled scattered factors, e.g. in shifting the $\tc$ in $\mathtt{acabb}$ or $\mathtt{aabcb}$ we are losing $\mathtt{bac}$ when shifting $\tc$ to the first position, $\mathtt{bca}$ when shifting into the middle, and $\mathtt{cba}$ when appending $\tc$.
\end{example}

\subsection{Equivalence testing for small $k$}

\begin{proposition}
    Let $k,\ell\in \N$. For all $w,v\in \Sigma^*$, we can decide whether $w \jsimkl v$ holds in time $O(|\Sigma|^{2k+3}k^2)$.
\end{proposition}
\begin{proof}
    Consider \Cref{alg:fixedkl} as an algorithm solving \Cref{problem:simon}.
    \begin{algorithm}[]\caption{Jumbled Simon's equivalence for small $k$}\label{alg:fixedkl}
\begin{algorithmic}
    \State Given: two strings $w,v \in \Sigma^*$, two integers~$k,\ell \in \N$
    \For{all strings $u \in \Sigma^{\leq k}$}
        \State $J\gets \emptyset$ \Comment{The set of strings obtained using $\le \ell$ jumbles on $u$}
        \For{all strings $u'$ of length at most $|u|$}
            \If{$p(u')\le p(u)$ and $\lcsf(u,u')\ge |u|-\ell$}
                add $u'$ to $J$
            \EndIf
        \EndFor
        \State $B_w\gets $ False, $B_v\gets$ False
        \For{all strings $u'\in J$}
            \If{$u'$ is a scattered factor of $w$}
                $B_w\gets$ True
            \EndIf
            \If{$u'$ is a scattered factor of $v$}
                $B_v\gets$ True
            \EndIf
        \EndFor
        \If{$B_w\ne B_v$}
            \Return False
        \EndIf
    \EndFor
    \State \Return True
\end{algorithmic}
\end{algorithm}

Correctness: we want to test $v\jsim_{k,\ell} w$ which is satisfied if each string $u$ of length $|u|\leq k$ is either $\ell$-jumbled scattered factor of both $v$ and $w$ or of neither one. Further, $u$ is $\ell$-jumbled in $v$ if and only if $p(u) \le p(v)$ and $\lcsf(u,v)\geq |u|-\ell$. Algorithm 1 tries to find a witness for $v\njsim_{k,\ell} w$ by simply testing for each string $u$ of length up to $k$ whether it is $\ell$-jumbled scattered factor of $v$ and not of $w$ or vice-versa. This is done by checking each possible string $u'$ whether it has a common scattered factor with $u$ of length at least $|u|-\ell$ while also being a scattered factor of $v$ or, respectively, $w$. Since $u'$ being a scattered factor of $v$ implies $\lcsf(v,u)\geq \lcsf(u',u)$ the existence of such string $u'$ witnesses $u$ being an $\ell$-jumbled scattered factor of $v$. Conversely, if $u$ is an $\ell$-jumbled scattered factor of $v$, by Definition~\ref{def:jidx}, there has to be a common scattered factor $u'$ of $u$ and $v$ of length $|u'|\geq |u|-\ell$, hence, $u'$ satisfies $\lcsf(u',u) = |u'| \geq |u|-\ell$ and $u'$ is a scattered factor of $v$. Accordingly, after each run of the outer for-loop $B_v$ correctly indicates whether $u$ is $\ell$-jumbled scattered factor of $v$, analogously $B_w$ correctly indicates whether $u$ is $\ell$-jumbled scattered factor of $w$. Now $u$ is a witness for $v\njsim_{k,\ell} w$ if an only if $B_v\neq B_w$. Finally, if none of the strings of length up to $k$ witnesses the non-equivalence of $v$ and $w$ then Algorithm $1$ correctly outputs that $u\jsim_{k,\ell} w$ is true.

Complexity:
for each $u$,
the construction of the set $J$ takes $O(|\Sigma|^{|u|+2}|u|^2)$ time,
since the computation of $\LCS(u,u')$ takes $O(|u||u'|)$ time
and the construction of the Parikh vector takes $O(|\Sigma|+|u|)$ time
for each choice of $u'$.
There are $O(|\Sigma|^{|u|+1})$ choices of $u'$,
and since $|u|\le k$, the bound follows.
Then, in the enumeration of the strings~$u'$ in $J$,
we spend $O(|w|+|v|+2|u'|)$ time for each choice of $u'$ in $J$.
Since there are at most $O(|\Sigma|^{|u|+1})$ strings in $J$,
we spend at most $O(|\Sigma|^{|u|+1}(|w|+|v|+k))$ in the enumeration 
of the strings in $J$.
Finally, there are at most $O(|\Sigma|^{|u|+1})$ choices of $u$;
therefore the overall running time of the algorithm
is $O(|\Sigma|^{2k+3}k^2)$.
\end{proof}

\begin{corollary}
    If $k$ is in
    $O(\log_{|\Sigma|}(|w|+|v|))$,
    then we can
    decide whether or not
    $w \jsimkl v$ holds
    in polynomial time $O(|\Sigma|^3(|w|+|v|)^2k^2)$.
\end{corollary}

\subsection{An automaton-based approach}

We present an automata approach for testing
$\ell$-jumbled $k$-Simon equivalence
between two strings.
However, the following result
makes it unlikely to have
an efficient solution by constructing an automaton
for all jumbled scattered factors of a given word for the case of general $k$ and $\ell$.

\begin{lemma}\label{prop:exp}
    Let $w \in \Sigma^*$ satisfying $\al (w) = \Sigma$.
    For $k,\ell \in \N$ satisfying
    $\ell= k = |w|$,
    the size of the minimal NFA recognising
    $\JScatFact_{\leq k}(w, \ell)$
    is exponential in the size of $\Sigma$.
\end{lemma}
\begin{proof}
    We prove a lower bound on the number of states
    of an NFA using a fooling set argument (see e.g. \cite{DBLP:journals/ipl/Birget92,DBLP:journals/ijfcs/HolzerK09}). Recall the fooling set argument: let $L \subseteq \Sigma^*$ be a regular language and suppose there exists a set of pairs $\mathcal F := \{(x_i,y_i) \mid 1 \leq i \leq n\}\subset \Sigma^\ast\times\Sigma^\ast$, called the \emph{fooling set}, such that
    \begin{enumerate}
        \item $x_i y_i \in L$ for $1 \leq i \leq n$, and
        \item $i \neq j$ implies $x_i y_j \notin L$ or $x_j y_i \notin L$ for $1 \leq i,j \leq n$.
    \end{enumerate}
    Then any NFA which recognises $L$ has at least $\len{\mathcal F}$ many states.
    
    
    For a subset~$S$ of $\Sigma$,
    let us define the word~$w(S)$ of length~$|S|$
    that satisfies $\len {w(S)}_{\ta}=1$
    whenever $\ta\in S$
    and $\len{w(S)}_{\ta}=0$ otherwise.
    Let us choose the string~$w=w(\Sigma)$
    and $\ell=k=\len w$.
    Then $\JSF_k(w,\ell)$ contains all permutations
    of $w$.
    Consider the fooling set
    that consists of the pairs~$(w(S),w(\Sigma\setminus S))$
    for each subset~$S$ of $\Sigma$.
    For any two distinct subsets~$S_1,S_2$
    of $\Sigma$,
    we prove that either $w(S_1)w(\Sigma\setminus S_2)$
    or $w(S_2)w(\Sigma\setminus S_1)$
    is not in the language.
    Without loss of generality,
    there exists a letter~$\ta\in \Sigma$
    that appears in $S_1$ but not in $S_2$.
    Then, the word~$w(S_1)w(\Sigma\setminus S_2)$
    contains two instances of $\ta$,
    and therefore it is not in $\JSF_k(w,\ell)$.
    Therefore, there are at least
    $2^{|\Sigma|}$ elements
    in a fooling set for $\JSF_k(w,\ell)$,
    and thus the smallest NFA
    recognizing $\JSF_k(w,\ell)$
    has size exponential in $|\Sigma|$.
\end{proof}

\Cref{prop:exp} implies that, for arbitrary $k,\ell$, an automaton which recognises all $\ell$-jumbled scattered factors of length at most $k$ of a word can have size $O(2^{\len\Sigma})$. Since the Parikh-vector can easily be checked algorithmically we focus on automata to accept the set of all words admitting a longest common scattered factor of sufficient length.
As a useful tool, we define the X-ranker~$\XYranker{X}{w}{i}{\ta}:=\min(\{j> i\mid w[j]=\ta\}\cup\{|w|+1\})$~\cite{SchwentickTV01} that outputs the next occurrence of the letter $\ta$ after the \nth{$i$} index in $w$ or $|w|+1$ if the letter is not occurring any more. For example, $\XYranker{X}{\mathtt{abaab}}{0}{\ta}=1$
since $\mathtt{abaab}[1]=\ta$,
and $\XYranker{X}{\mathtt{abaab}}{1}{\ta}=3$
since $\mathtt{abaab}[3]$ is the next occurrence of $\ta$
after the first $\ta$.

For the first approach, we characterise non-equivalent words by their longest common scattered factors.
For a string~$u$ and an integer~$k$,
let $p(u,k)$ be the vector of size~$|\Sigma|$
where $p(u,k)[\sigma]=\min(k,|u|_\sigma)$.

\begin{lemma}\label{lem:simklCharacterize}
    For $w_1,w_2\in \Sigma^*$
    and two integers~$k,\ell \in \N$,
    we have $w_1 \not \jsimkl w_2$
    if and only if $p(w_1,k)\ne p(w_2,k)$,
    or there exists a string~$u$ of length at most~$k$
    with either $\lcsf(u,w_1)<|u|-\ell\le \lcsf(u,w_2)$
    or $\lcsf(u,w_2)<|u|-\ell\le \lcsf(u,w_1)$.
\end{lemma}
\begin{proof}
    ``$\Leftarrow$'': Let us first consider the case
    where $p(w_1,k)\ne p(w_2,k)$.
    Choose a character~$\ta$ with $p(w_1,k)[\ta]\ne p(w_2,k)[\ta]$
    and let $v=\ta^{x}$
    where $x$ is the maximum between
    $p(w_1,k)[\ta]$ and $p(w_2,k)[\ta]$.
    Note that $|v|\le k$,
    and only one of $w_1$ and $w_2$ contains $v$
    as an $\ell$-jumbled scattered factor.
    Therefore, we have $w_1 \not \jsimkl w_2$.
    For the other case when
    there exists a string~$u$ of length at most~$k$
    with either $\lcsf(u,w_1)<|u|-\ell\le \lcsf(u,w_2)$
    or $\lcsf(u,w_2)<|u|-\ell\le \lcsf(u,w_1)$,
    let us assume without loss of generality
    that $\lcsf(u,w_1)<|u|-\ell\le \lcsf(u,w_2)$.
    By the definition of the jumbled scattered factor
    of a string,
    $u$ is an $\ell$-jumbled scattered factor of $w_2$
    but not for $w_1$.
    Therefore, we have $w_1 \not \jsimkl w_2$.

    ``$\Rightarrow$'': Let $w_1 \not \jsimkl w_2$;
    then there exists a string~$u$ of length at most $k$
    that is an $\ell$-jumbled scattered factor of one string
    but is not of the other.
    Without loss of generality,
    let $u$ be an $\ell$-jumbled scattered factor of $w_2$;
    then we have $p(u)\le p(w_2)$.
    First, consider the case when
    $|u|_\ta > |w_1|_\ta$
    for some $\ta\in\Sigma$.
    Then, since $p(u)[\ta]\le k$
    and $p(u)[\ta] \le p(w_1)[\ta]$,
    we have
    $p(w_1,k)[\ta]\ne p(w_2,k)[\ta]$.
    Otherwise,
    we have $|u|_\ta \le |w_1|_\ta$
    for all $\ta\in\Sigma$.
    Then, we have $p(u)\le p(w_1,k)\le p(w_1)$.
    In other words, for $u$ to not be an $\ell$-jumbled scattered factor
    of $w_1$,
    the length of the longest common scattered factor of $u$ and $w_1$
    should be of length strictly less than $|u|-\ell$.
    Since $u$ is an $\ell$-jumbled scattered factor of $w_2$,
    we have $|u|-\ell \le \lcsf(u,w_2)$,
    and thus the claim holds.
\end{proof}

We now use \cref{lem:simklCharacterize} for constructing a DFA dependent on some word $w \in \Sigma^*$ that accepts all words of length at most $k$ that share a common scattered factor of length $k-\ell$ with $w$. The idea is to compare the DFAs for different words to decide equivalence. Therefore, we define the set of words that share a longest common scattered factor of length at least $k-\ell$ with a given word $w$.

\begin{definition}
    For $w \in \Sigma^*$ and $k,\ell \in \N_0$, we define $\LCS(w,k,\ell):=\{u\in\Sigma^{k}\mid \lcsf(u,w)\ge k-\ell\}$.
\end{definition}

Given a string~$w$
and two integers~$k$ and $\ell$,
we construct a DFA~$A$ that recognises
$\LCS(w,k,\ell)$.
The idea is to simulate the classic LCS dynamic programming algorithm (cf., e.g. \cite{DBLP:books/cu/Gusfield1997})
one row at a time,
where each row corresponds to a character
of the input
and each column corresponds to a character in $w$.
We leverage two important observations.
First, the values of each cell can grow up to $k-\ell$,
and if any value exceeds $k-\ell$,
then we should automatically accept the input.
Second, the values in each row
are non-decreasing.
Therefore, it suffices to remember
a non-decreasing sequence of $k-\ell$ positions of $w$,
where the $i^\text{th}$ position
is the index of the smallest cell
in the current row of the DP table
with value greater than $i$.
We formally construct the automaton~$A:=(Q_A,\Sigma,\delta_A,s_A,F_A)$ as follows:
\begin{itemize}
    \item $Q_A:=[k]_0\times \{\vect{V}\in ([|w|+1]_0)^{k-\ell}\mid
    \forall i\in [k-\ell]:\vect{V}[i-1]\le \vect{V}[i] \}$,
    \item $s_A:=(0,[0,|w|+1,\ldots,|w|+1])$,
    \item for $x < k$ and $\vect{V}\in ([|w|+1]_0)^{k-\ell}$,
    the transition function is defined as~$\delta_A((x,\vect{V}),\ta):=(x+1,\vect{U})$,
    where $\vect{U}[0]:=0$ and 
    \[\vect{U}[i]:=\min(\vect{V}[i],\XYranker{X}{w}{\vect{V}[i-1]}{\ta})\]
    for all $i\in [k-\ell]$,
    \item $F_A:=\{(k,\vect{V})\in Q_A\mid \vect{V}[k-\ell]\le |w|\}$.
\end{itemize}
In other words, each state is a pair
where the first item is a counter that grows up to $k$,
and the second item is a non-decreasing sequence of positions
representing a row of the DP table.
Here, $\vect{V}[i]=|w|+1$ indicates that there is no cell
in the current row
with value at least $i$.
Then, it is easy to see that $(0,[0,|w|+1,\ldots,|w|+1])$
represents the first row of the DP table.
At each input,
we increment the counter if it is less than $k$;
we reject if the counter is already $k$.
Then, we update the vector accordingly
so that $\vect{U}[i]$ correctly contains
the smallest index of the DP table
with value at least $i$.
Finally, we accept if the counter is exactly $k$,
and the $(k-\ell)^\text{th}$ cell
has a value other than $|w|+1$,
which means there is a cell with value~$k-\ell$
in the DP table.
Therefore, we have the following lemma.

\begin{theorem}
    There exists a DFA
    of size~$(k+1)(|w|+1)^{k-\ell}$
    that recognises $\LCS(w,k,\ell)$.
\end{theorem}

\begin{corollary}\label{cor:decideDiffConst}
    If $k-\ell\le c$ for some constant~$c$,
    we can decide whether or not
    $w_1 \jsimkl w_2$
    in polynomial time.
\end{corollary}
\begin{proof}
    We first test whether or not
    $p(w_1,k)=p(w_2,k)$.
    If not, we return false.
    Otherwise, for each $x\le k$,
    we construct DFAs
    that recognise $\LCS(w_1,x,\ell)$
    and $\LCS(w_2,x,\ell)$, respectively.
    If the two DFAs are not equivalent,
    we return false.
    If the pairs of DFAs are equivalent
    for all $x\le k$,
    then we return true.
    The correctness of this procedure
    follows from Lemma~\ref{lem:simklCharacterize}.
    Furthermore, since $x-\ell \le k-\ell\le c$,
    the size of the DFA is $O(kn^c)$, 
    where $n$ is the length of $w_1$ or $w_2$.
    Therefore, the size of the DFA is polynomial in $k$ and $n$,
    and since the equivalence of two DFAs can be checked
    in polynomial time,
    the entire procedure takes polynomial time.
\end{proof}

We now modify the DFA construction
based on the intuition
that the length of the longest common scattered factor
between $w$ and the suffix of another string~$u$
is at most the length of the suffix.
In other words, in constructing $\LCS(w,k,\ell)$,
we do not need to consider
scattered factors that start
after reading $\ell$ characters,
because there are only $k-\ell$ remaining characters to read.
In this sense, we construct a new automaton~$B:=(Q_B,\Sigma,\delta_B,s_B,F_B)$ as follows.
\begin{itemize}
    \item $Q_B:=[k]_0\times \{\vect{V}\in ([|w|+1]_0)^{\ell}\mid
    \forall i\in [\ell]:\vect{V}[i-1]\le \vect{V}[i] \}$,
    \item $s_B:=(0,[0,|w|+1,\ldots,|w|+1])$,
    \item for $x < k$ and $\vect{V}\in ([|w|+1]_0)^{\ell}$,
    the transition function is defined as
    \[\delta_B((x,\vect{V}),\ta):=(x+1,\vect{U}),\text{~where}\]
    \begin{itemize}
        \item if $x<\ell$, then $\vect{U}[0]:=0$ and 
        \[\vect{U}[i]:=\min(\vect{V}[i],\XYranker{X}{w}{\vect{V}[i-1]}{\ta})
        \text{~for~all~}i\in [\ell],\]
        \item otherwise, $\vect{U}[\ell-1]:=\XYranker{X}{w}{\vect{V}[\ell-1]}{\ta}$ and
        \[\vect{U}[i]:=\min(\vect{V}[i+1],\XYranker{X}{w}{\vect{V}[i]}{\ta})\]
        for all $i\in [\ell-1]$,
    \end{itemize}
    \item $F_B:=\{(k,\vect{V})\in Q_B\mid \vect{V}[0]\le |w|\}$.
\end{itemize}
The idea is to simulate the automaton~$A$
using a window of size~$\ell$
on the non-decreasing sequence of positions;
in the first $\ell$ steps,
this window captures
the first $\ell$ elements,
and in the later steps,
say $(\ell+j)^\text{th}$ step,
this window captures
the sequence of $\ell$ elements
starting from the $j^\text{th}$ element.
Therefore, at the $k^\text{th}$ step,
the $(k-\ell)^\text{th}$ element in the DFA~$A$
is the $0^\text{th}$ element in the DFA~$B$,
and thus we accept if $\vect{V}[0]\le |w|$.

\begin{theorem}
    There exists a DFA
    of size~$(k+1)(|w|+1)^{\ell}$
    that recognises $\LCS(w,k,\ell)$.
\end{theorem}

\begin{corollary}\label{cor:decideEllConst}
    If $\ell$ is constant,
    we can decide whether 
    $w_1 \jsimkl w_2$ holds
    in polynomial time.
\end{corollary}
\begin{proof}
    We use exactly the same procedure
    as Corollary~\ref{cor:decideDiffConst},
    but this time we use the DFA~$B$
    instead of the DFA~$A$.
    Then, the size of the DFA
    is at most $O(kn^\ell)$,
    which is polynomial in $k$ and $n$.
    Therefore, we can decide whether or not
    $w_1 \jsimkl w_2$
    in polynomial time
    if $\ell$ is constant.
\end{proof}

Unfortunately, when $\ell$ and $k-\ell$
are arbitrary,
it is difficult to construct a DFA
that accepts all strings
in $\LCS(w,k,\ell)$.
Instead, we present an NFA
that accepts $\LCS(w,k,\ell)$ in the following.
Let the NFA $\mathcal A$ be defined by
\begin{itemize}
    \item $Q := [|w|]_0 \times [\ell]_0$
    \item $s := (0,0)$
    \item $F := Q$, and 
    \item for all $b \le \ell, i\le |w|$ define 
    \[
    \delta((i,b),\ta):=\begin{cases}
        \{(i,b+1), (X(w,i,\ta),b)\}, & \text{if } X(w,i,\ta) \neq |w|+1\text{ and } b<\ell,\\
        \{(X(w,i,\ta),b)\}, & \text{if } X(w,i,\ta) \neq |w|+1\text{ and } b=\ell,\\
        \{(i,b+1)\}, & \text{if } X(w,i,\ta) = |w|+1\text{ and } b<\ell,\\
        \emptyset, & \text{otherwise}.
    \end{cases}
    \]
\end{itemize}

Note that
an equivalent DFA can be obtained
by applying the powerset construction
on $\mathcal A$.
For an examplary construction of the NFA for $w = \mathtt{abb}$ and $\ell = 1$, consider \Cref{fig:automaton}.

\begin{figure}
    \centering
    \begin{tikzpicture}[shorten >=1pt,auto,node distance=2.5cm,
			semithick]
			\tikzstyle{every state}=[initial text=$ $]
			\node[initial,state,accepting] (00)                    {$0,0$};
			\node[state,accepting]         (10) [right of=00]      {$1,0$};
            \node[state,accepting]         (20) [right of=10]      {$2,0$};
            \node[state,accepting]         (30) [right of=20]      {$3,0$};

            \node[state,accepting]         (01) [below of=00]      {$0,1$};
			\node[state,accepting]         (11) [right of=01]      {$1,1$};
            \node[state,accepting]         (21) [right of=11]      {$2,1$};
            \node[state,accepting]         (31) [right of=21]      {$3,1$};

			
			\path[->] (00) edge[left] node {$\ta,\tb$} (01)
			          edge node {$\mathtt{a}$} (10)
			          edge[bend left] node {$\mathtt{b}$} (20)
            (10) edge[left] node {$\ta,\tb$} (11)
                 edge node {$\mathtt{b}$} (20)
            (20) edge[left] node {$\ta,\tb$} (21)
                 edge node {$\mathtt{b}$} (30)
            (30) edge[left] node {$\ta,\tb$} (31)
            
            (01) 
			          edge node {$\mathtt{a}$} (11)
			          edge[bend right] node[below] {$\mathtt{b}$} (21)
            (11) 
                 edge node {$\mathtt{b}$} (21)
            (21) 
                 edge node {$\mathtt{b}$} (31)

            ;
		\end{tikzpicture}
    \caption{The NFA for recognising the $1$-jumbled scattered factors of $w = \mathtt{abb}$. 
    }
    \label{fig:automaton}
\end{figure}
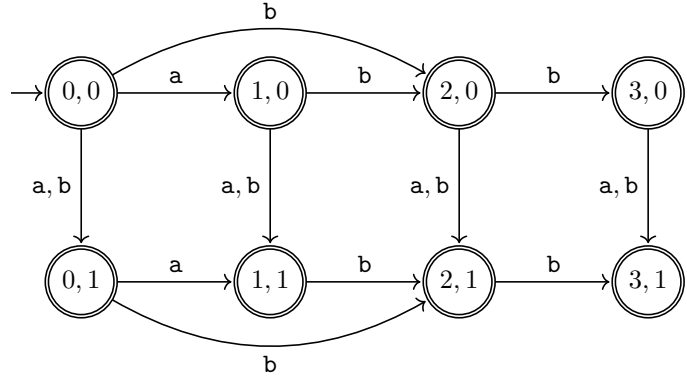

%% file: binom.tex
Research on scattered factors has been successful in finding formulas and identities for counting the different ways a scattered factor can be found in a word (multiplicities). For classical scattered factors the binomial coefficient is defined as the number of embeddings of $u$ in $w$ as a scattered factor (with no jumbles), i.e. $\binom{w}{u} = |E(u,w)|$ where $E(u,w)$ denotes the set of embeddings of $u$ in $w$. In this section, we want to expand such efforts to jumbled scattered factors.

First though, we have to clarify what exactly we mean by multiplicities. This leads us to the jumbled binomial coefficient.

\begin{definition}[Jumbled Binomial Coefficient]
  Let $u, w \in \Sigma^{*}$. Then, the \emph{jumbled binomial coefficient} $\jbinom{w}{u}$ is defined as \(
  \jbinom{w}{u} \mathrel{:=} \left|\loosejumbledembeddings(u, w)\right|.
\) Further, we also define the \emph{$\ell$-jumbled binomial coefficient} as $\jbinom[\ell]{w}{u} \mathrel{:=} \left|\loosejumbledembeddings(u, w,\ell)\right|$. \end{definition}

Therefore, the number of ways to find a jumbled scattered factor in a word is the number of jumbled embeddings. For example, let \(
u = \mathtt{an}, w = \mathtt{anna}
\), then $\jbinom{w}{u} = 4$ since we can find the following embeddings (cf. Figure~\ref{fig:embeddings-an-to-anna}).

    
\begin{figure}[h]
    \centering
\begin{tikzpicture}

  \node (a_top) at (0,0) {\texttt{a}};
  \node (n_top) [right=0.1cm of a_top] {\texttt{n}};

  \node (a_1) at (0,-1.5) {\texttt{a}};
  \node (n_1) [right=0.1cm of a_1]{\texttt{n}};
  \node (n_2) [right=0.1cm of n_1] {\texttt{n}};
  \node (a_2) [right=0.1cm of n_2] {\texttt{a}};

   \draw[->, thick] (a_top.south) -- (a_1.north);
  \draw[->, thick] (n_top.south) -- (n_1.north);

    \node (a_top) at (3,0) {\texttt{a}};
  \node (n_top) [right=0.1cm of a_top] {\texttt{n}};

  \node (a_1) at (3,-1.5) {\texttt{a}};
  \node (n_1) [right=0.1cm of a_1]{\texttt{n}};
  \node (n_2) [right=0.1cm of n_1] {\texttt{n}};
  \node (a_2) [right=0.1cm of n_2] {\texttt{a}};

   \draw[->, thick] (a_top.south) -- (a_1.north);
  \draw[->, thick] (n_top.south) -- (n_2.north);

    \node (a_top) at (6,0) {\texttt{a}};
  \node (n_top) [right=0.1cm of a_top] {\texttt{n}};
 
  \node (a_1) at (6,-1.5) {\texttt{a}};
  \node (n_1) [right=0.1cm of a_1]{\texttt{n}};
  \node (n_2) [right=0.1cm of n_1] {\texttt{n}};
  \node (a_2) [right=0.1cm of n_2] {\texttt{a}};

   \draw[->, thick, blue] (a_top.south) -- (a_2.north);
  \draw[->, thick] (n_top.south) -- (n_1.north);
  
    \node (a_top) at (9,0) {\texttt{a}};
  \node (n_top) [right=0.1cm of a_top] {\texttt{n}};
 
  \node (a_1) at (9,-1.5) {\texttt{a}};
  \node (n_1) [right=0.1cm of a_1]{\texttt{n}};
  \node (n_2) [right=0.1cm of n_1] {\texttt{n}};
  \node (a_2) [right=0.1cm of n_2] {\texttt{a}};

   \draw[->, thick, blue] (a_top.south) -- (a_2.north);
  \draw[->, thick] (n_top.south) -- (n_2.north);
\end{tikzpicture}
    \caption{Four (jumbled) embeddings of \texttt{an} in \texttt{anna} with a blue arrow indicating jumbles.}
    \label{fig:embeddings-an-to-anna}
\end{figure}
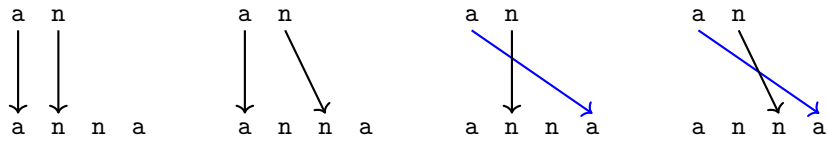
We can also immediately see that $\jbinom[0]{w}{u}=\binom{w}{u} =2$ since only the first two embeddings are non-jumbled.

\begin{proposition}
  Let $w \in \Sigma^{*}$ and $u \in \loosejscatfact(w)$. Then,
  \begin{equation*}
  \jbinom{w}{u} = \prod_{\ta \in \Sigma} \frac{\len{w}_\ta!}{(\len{w}_\ta - \len{u}_\ta)!}.
  \end{equation*}
\end{proposition}

\begin{proof}
Since we do not restrict the jumbles,
we just count all possible
injective mappings~$e : [\len{u}] \rightarrow [\len{w}]$
with $u[i] = w[e(i)]$ for all $i \in [\len{u}]$.
If $u = \eps$ then there is only the trivial mapping,
so $\jbinom{w}{u} = 1$. For the case $u \neq \eps$,
we first note that any $i , j \in [\len{u}]$ with $u[i] \neq u[j]$ 
will be mapped independently by any such mapping,
since we know $e(i) \neq e'(j)$
for all jumbled embeddings~$e , e'$.
Now consider sets $I_{\ta} \mathrel{:=} \left\{i \in [\len{u}] \mid u[i] = \ta\right\}$ and $J_{\ta} \mathrel{:=} \left\{j \in [\len{w}] \mid w[j] = \ta\right\}$ for any $\ta \in \Sigma$. Then, for any jumbled embedding $e : [\len{u}] \rightarrow [\len{w}]$ we know that there are \(
\frac{\left|J_{\ta}\right|!}{(\left|J_{\ta}\right| - \left|I_{\ta}\right|)!} = \frac{\len{w}_{\ta}!}{(\len{w}_{\ta} - \len{u}_{\ta})!}
\) injective mappings $\restrict{e}{I_{\ta}} : I_{\ta} \rightarrow J_{\ta}$. The formula follows immediately from the independence of $\restrict{e}{ I_{\ta}}$ and $\restrict{e}{ I_{\tb}}$ for any $e \in \loosejumbledembeddings(u, w)$, $\ta , \tb \in \Sigma$ with \(
\ta \neq \tb
\). \end{proof}

\begin{corollary}\label{cor:unique}
If $w\in\Sigma^{\ast}$ with $|\al(w)|=|w|$,
then $|\loosejumbledembeddings(u,w,\ell)|=1$ for all $u\in\JScatFact(w,\ell)$.
Further, if $w \in \Sigma^*$ and for all $\ta\in\al(u)$ we have $|w|_{\ta} = |u|_{\ta}=1$,
then $|\loosejumbledembeddings(u,w,\ell)| = 1$ for all $u \in \JScatFact(w,\ell)$.
\end{corollary}

We list some example calculations, but we do not list all possible embeddings explicitly.
\[
\jbinom{\mathtt{anna}}{\mathtt{an}} = \left(\frac{2 !}{(2 - 1)!}\right)^2 = 4 \quad \quad \quad \jbinom{\mathtt{banana}}{\mathtt{ana}} = \frac{3 !}{(3 - 2)!} \cdot \frac{1 !}{(1 - 0)!} \cdot \frac{2 !}{(2 - 1)!} = 12
\]

Notably, it becomes more complex when we restrict the number of jumbles. 
One strategy to compute all the $\ell$-jumbled embeddings and thus the $\ell$-jumbled binomial coefficient is to
\begin{enumerate}
    \item enumerate all jumbled embeddings, ignoring the restrictions on the number of jumbles,
    \item check the amount of jumbles by checking the longest increasing common scattered factor.
\end{enumerate}
This strategy also leads us to the following recurrence for the $\ell$-jumbled binomial coefficient.

\begin{proposition}
Given $w\in\Sigma^*$, an integer \(\ell \in \N_0 \), and $u\in\JSF(w)$,
let $\Gamma := \{\tc_i\mid 1\le i\le\len u\}$ be an alphabet
disjoint from $\Sigma$
that satisfies $\tc_i<\tc_j$
whenever $i<j$.
Then, the $\ell$-jumbled binomial coefficient $\jbinom[\ell]{w}{u}$ can be computed by the recursion
\begin{align}
    \jbinom[\ell]{w}{u} &=  \operatorname{LIS}(w, \len{u}-\ell)
     &&\text{if } u = \tc_1\cdots \tc_{\len u} \text{ and}
    \label{eq:jbinom2}\\
  \jbinom[\ell]{w}{u} &= \sum_{w = x' \ta y'} \jbinom[\ell]{x' \tc_{\len{x \ta}} y'}{x \tc_{\len{x \ta}} y}
   &&\text{otherwise, where }x\ta y=u
  \text{ with }x\in \Gamma^*, \ta\in \Sigma
  \label{eq:jbinom1}
\end{align}
where $\operatorname{LIS}(w,m)=1$
if the length of the longest increasing subsequence of $w$
among the characters in $\Gamma$
is at least $m$
and $0$ otherwise.
\end{proposition}
 \label{prop:reccurence}
\begin{proof}
Let $u , w \in \Sigma^{*}$ and $u' = \tc_1 \cdots \tc_{\len{u}}$. 
We prove the following statements.
 \begin{enumerate}
 \item Let $u'$ as above and \(w' \in (\Sigma \cup \Gamma)^{*} \) be the words that we get when \cref{eq:jbinom2} applies. Then, $\len{w'}_{\tc_i} = 1$ for all \(i \in [\len{u}] \) and there is one unique embedding \(e \in \loosejumbledembeddings(u', w') \) and $e \in \loosejumbledembeddings(u, w)$.
   \item Let $e \in \loosejumbledembeddings(u, w)$ and $\len{w} = \len{w'}$ with \[
w'[i] = \begin{cases}
 \tc_{j}, \text{if there exists } j \text{ such that }e(j)=i \\
 w[i], \text{otherwise}
\end{cases}
\] for all $i \in [\len{w}]$. Then, the computation of $\jbinom[\ell]{w}{u}$ will reach $\jbinom[\ell]{w'}{u'}$. 
    \item Assume that, after replacing every letter of $u$, we arrive at $\jbinom[\ell]{w}{u} = \sum_{i=1}^m \jbinom[\ell]{w_i}{u'}$ for some $m\in\mathbb N$. Then, $w_i\neq w_{i'}$ for all $i\neq i'$.
\end{enumerate}

\emph{Ad 1.} By construction, every $\tc_i$ is assigned only once, thus $\len{w'}_{\tc_i} = 1$. For the same reason $u' = \tc_1 \ldots \tc_{\len{u}}$. By \cref{cor:unique}, we get that there is a unique embedding $e \in \loosejumbledembeddings(u', w')$. Since $\tc_i$ is only assigned to the same letter in $u$ and $w$, we also get $e \in \loosejumbledembeddings(u, w)$.

\emph{Ad 2.} Let $e \in \loosejumbledembeddings(u, w)$. Then, by \cref{eq:jbinom1} we can choose to assign the $i^\text{th}$ letter of $u$ to the $e(i)^\text{th}$ letter of $w$ since $u[i] = w[e(i)]$. Doing this for all letters of $u$ gives said words. 

\emph{Ad 3.} Let $w=w_1\ta w_2\ta\cdots \ta w_{\len{w}_\ta+1}$ with $w_i\in(\Sigma\setminus\{\ta\})^\ast$ and $u=\ta y$. After the first recursion there will be $\len{w}_\ta$ many summands $\jbinom[\ell]{w'_i}{\tc_1y}$ where $w'_i$ is obtained from $w$ by replacing the $i^\text{th}$ occurrence of $\ta$ by $\tc_1$. Hence the statement holds after the first step. Now assume after $j$ steps we have $\jbinom[\ell]{w}{u}=\sum_{i=1}^{m_j}\jbinom[\ell]{w_i}{\tc_1\cdots\tc_j\ta y}$, for some $m_j\in\mathbb N$ and $w_i\neq w_{i'}$ for all $i\neq i'$. Then for any $i$ the coefficient $\jbinom[\ell]{w_i}{\tc_1\cdots\tc_j\ta y}$ will split up into unique summands with the same argument as in the base case. Furthermore, since $w_i$ and $w_{i'}$ differ in the positions of the letters $\tc_1,\ldots,\tc_j$, which are untouched be the $j+1^\text{th}$ recursion step, every summand derived from $w_i$ differs in the positions of $\tc_1,\ldots,\tc_j$ from every summand derived from $w_{i'}$. Statement $3$ follows by induction.

Now, by Statement $3$, $\jbinom[\ell]{w}{u}$ splits up into several unique coefficients $\jbinom[\ell]{w'}{u'}$, each of them corresponds to one jumbled embedding of $u$ into $w$ by Statement $2$ and only one of those embeddings by Statement $1$. Lastly, $\jbinom[\ell]{w'}{u'}=1$ if and only if the embedding of $u'$ into $w'$ has at most $\ell$ jumbles. Accordingly, the recursion correctly computes the amount of jumbled embeddings of $u$ into $w$.
\end{proof}

\subsection{Decompositions of $\ell$-Jumbled Embeddings}

Before we continue, we will note two interesting properties of the jumbled binomial coefficients and investigate the behaviour of $\ell$-jumbled embeddings more closely.

\begin{lemma}\label{lem:bijective-embedding}
Let $u , w \in \Sigma^{*} , \ell \in \N_0$ and \(
e \in \loosejumbledembeddings(u, w,\ell)
\). If $\len{u} = \len{w}$, then $e$ is bijective and $e^{-1} \in \loosejumbledembeddings(w, u,\ell)$. \end{lemma}
\begin{proof}
Let $L \subset [\len{u}]$ with $\left|L\right| \leq \ell$ such that $\restrict{e}{ [\len{u}] \setminus L}$ increasing. By definition $e$ is injective. Thus, $\len{u} = \len{w}$ implies that $e$ is bijective. Since $\restrict{e}{ [\len{u}] \setminus L}$ is increasing, so is $\restrict{e^{-1}}{ [\len{u}] \setminus L}$. Also, \(u[i] = u[e^{-1}(e(i))] = w[e(i)] \). \end{proof}

\begin{corollary}
Let $u , w \in \Sigma^{*}$ and $\ell \in \N_0$. If $\len{u} = \len{w}$, then \[
\jbinom[\ell]{w}{u} = \jbinom[\ell]{u}{w}.
\] \end{corollary}

\begin{proof}
This follows immediately from \cref{lem:bijective-embedding} and the fact that $e \neq e'$ iff $e^{-1} \neq e'^{-1}$ for any \(e , e' \in \loosejumbledembeddings(u, w,\ell) \). \end{proof}

\begin{lemma}\label{lem:jumbled-embedding-factorisation}
Let $u , w \in \Sigma^{*} , \ell \in \N_0$ and \(\widetilde{e} \in \loosejumbledembeddings(u, w,\ell) \). 
 \begin{enumerate}
 \item Then, we can factor \(\widetilde{e} = e \circ \widetilde{f} \) such that $e \in \embeddings(v, w)$ and \(\widetilde{f} \in \loosejumbledembeddings(u, v,\ell) \) for any $v \in \ScatFact_{\len{u}}(w)$ with $e([\len{v}]) \subset \widetilde{e}([\len{u}])$. \item In particular, the factorisation is unique if $\len{v} = \len{u}$.
 \item Conversely, given any \(v \in \Sigma^{*} , e \in \embeddings(v, w), \widetilde{e} \in \loosejumbledembeddings(u, v,\ell) \), the composition \(e \circ \widetilde{e} \in \loosejumbledembeddings(u, w,\ell) \).
 \end{enumerate}
 \end{lemma}

\begin{proof}
\emph{Ad 1.} Let $u , w \in \Sigma^{*} , \ell \in \N_0$ and \(\widetilde{e} \in \loosejumbledembeddings(u, w) \). Let $v \in \ScatFact_{\geq \len{u}}(w)$ and $e \in \embeddings(v, w)$ be an embedding with \(e([\len{v}]) \subset \widetilde{e}([\len{u}]) \). Then, \(\widetilde{f} : [\len{u}] \rightarrow [\len{v}] \), $j \mapsto e^{-1}(\widetilde{e}(j))$ is a well-defined jumbled embedding and $\widetilde{e} = e \circ \widetilde{f}$. Finally, for any set $S \subset [\len{u}]$ with $\restrict{e \circ \widetilde{f}}{ S}$ increasing, we know that $\restrict{\widetilde{f}}{ S}$ must be increasing since $e$ is increasing.

\emph{Ad 2.} Now, let $\len{v} = \len{u}$ and \(\widetilde{e} = e' \circ \widetilde{f}' \) be another factorisation with \(e' \in \embeddings(v, w), \widetilde{f} \in \loosejumbledembeddings_\ell(u, w) \). From Lemma~\ref{lem:bijective-embedding} we know that $\widetilde{f} , \widetilde{f}'$ are bijective. Thus, $\operatorname{image}(e) = \operatorname{image}(\widetilde{e}) = \operatorname{image}(e')$ and \(e = e' \) since $e , e'$ are increasing and injective. Now, \(e = e' \circ \widetilde{f}' \circ \widetilde{f}^{-1} = e' \) and thus \(\widetilde{f}' \circ \widetilde{f}^{-1} = id_{[\len{v}]} \Leftrightarrow \widetilde{f}' = \widetilde{f} \) since $e'$ is injective.

\emph{Ad 3.} Finally, let \(v \in \Sigma^{*} , e \in \embeddings(v, w), \widetilde{e} \in \loosejumbledembeddings(u, v) \) and \(\widetilde{e}' \mathrel{:=} e \circ \widetilde{e} \). Then, $\widetilde{e}'$ is injective as the composition of injective functions and for any $S \subset [\len{u}]$ we know that $\restrict{\widetilde{e}}{ S}$ increasing implies $\restrict{\widetilde{e}'}{ S}$ is increasing as the composition of increasing functions. Further, $u[i] = v[\widetilde{e}(i)] = w[e(\widetilde{e}(i))] = w[(e \circ \widetilde{e})(i)]$.
\end{proof}

\begin{remark}
Notably, by the proof of \cref{lem:jumbled-embedding-factorisation}, there is also only one $v$ such that we get a unique factorisation since $e([\len{v}]) = \widetilde{e}([\len{u}])$. From now on, we will denote this word by $\embeddingword(\widetilde{e})$. \end{remark}

\begin{theorem}\label{prop:split-on-length}
Let $u , w \in \Sigma^{*}$ and $\ell \in \N_0$. Then \[
\jbinom[\ell]{w}{u} = \sum_{v \in \ScatFact_{\len{u}}(w)} \binom{w}{v} \jbinom[\ell]{v}{u}.
\] 
\end{theorem}
\begin{proof}
By \cref{lem:jumbled-embedding-factorisation}, any jumbled embedding $\widetilde{e} \in \loosejumbledembeddings(u, w,\ell)$ factorises uniquely over $\embeddingword(\widetilde{e})$. Since \(\embeddingword(\widetilde{e}) \in \ScatFact_{\len{u}}(w) \), this implies \(\jbinom[\ell]{w}{u} \leq \sum_{v \in \ScatFact_{\len{u}}(w)} \binom{w}{v} \jbinom[\ell]{v}{u} \). On the other hand, any $e \in \embeddings(v, w)$ and \(
\widetilde{e} \in \loosejumbledembeddings(u, v,\ell)
\) yields a jumbled embedding \(
e \circ \widetilde{e} \in \loosejumbledembeddings(u, w,\ell)
\), so \(
\jbinom[\ell]{w}{u} \geq \sum_{v \in \ScatFact_{\len{u}}(w)} \binom{w}{v} \jbinom[\ell]{v}{u}
\).
\end{proof}

\subsection{A Connection to Young Tableaux}

In 1961 Schensted related permutations to Young tableaux \cite{Schensted_1961}. A Young tableaux, which is an arrangement of $n$ boxes labelled with numbers, has a shape $\mu=\mu(n,\ell)$. The shape represents a partition of $n$ and the largest element of this partition is called the width $\ell$ of the shape. In his work Schensted showed that there is a one-to-one correspondence between permutations of the numbers $1,2,\ldots,n$ and pairs of standard Young tableaux that have the same shape with $n$ boxes. Moreover, the width of the shape is exactly the length of the longest increasing subsequence of the permutation represented. 
Since Young tableaux are well-researched we spend the rest of this section to illuminate the connection between Young tableaux and the jumbled binomial coefficient. To be more precise, we consider the \emph{hook length formula} which uses the hook lengths $h_j^{\mu(n,\ell)}$ (see \cite{FRT54}) to obtain the following result.

\begin{proposition}\label{lem:count-lis-permutations} 
Let $\ell , n \in \mathbb{N}$. Then, the number of permutations with length $n$ and longest increasing subsequence of length $\ell$ is \[P(n ,\ell) :=
\sum_{\substack{\mu(n,\ell)}} \left(\frac{n !}{\prod_{j = 1}^n h_j^{\mu(n,\ell)}}\right)^2
\] where $\mu(n,\ell)$ is a shape having size $n$ and width $\ell$ and $h_j^{\mu(n,\ell)}$ are the hook lengths of $\mu(n,\ell)$. 
\end{proposition}
\begin{proof}
    Let $H_{\mu(n,\ell)} = \prod_{j=1}^n h_j^{\mu(n,\ell)}$. The hook length formula \cite{FRT54} states the  number of standard Young tableaux of shape $\mu(n,\ell)$ is given by $\frac{n!}{H_{\mu(n,\ell)}}$.
    By Schensted's correspondence \cite{Schensted_1961}, we know that permutations correspond to pairs of standard Young tableaux and thus there are $\left(\frac{n!}{H_{\mu(n,\ell)}}\right)^2$ permutations that yield different pairs of Young tableaux of shape $\mu(n,\ell)$ under repeated application of Schensted's operation. Now, we know that the length of the longest increasing subsequence is exactly the width of the shape. So, the sum over all shapes with width $\ell$ yields the number of permutations with length $n$ and longest increasing subsequence $\ell$.
\end{proof}

We know that the jumbled binomial coefficient $\jbinom[\ell]{w}{u}$ counts a restricted set of permutations with longest increasing subsequence of length at least $\len{u} - \ell$. This, leads us to the following corollary.

\begin{corollary}
Let $\ell \in \N_0$, $n , m \in \mathbb{N}$ and $\ta \in \Sigma$. Then, for $u \mathrel{:=} \ta^n$ and \(
w \mathrel{:=} \ta^m
\) we get \begin{align*}
\jbinom[\ell]{w}{u} = \binom{w}{u} \cdot \sum_{\ell' \in [\ell]_0} P(n , n -\ell') = \binom{m}{n} \cdot \sum_{\ell' \in [\ell]_0} P(n , n -\ell').
\end{align*}
\end{corollary}
\begin{proof}
By \cref{prop:split-on-length}, we know \(\jbinom[\ell]{w}{u} = \binom{w}{u} \jbinom[\ell]{u}{u} \) since either $\binom{w}{u} = 0$ or $u$ is the only scattered factor of $w$. Now, $\jbinom[\ell]{u}{u}$ counts exactly the bijective mappings $[n] \rightarrow [n]$ with longest increasing subsequence of length at most $n -\ell$ which is exactly \(\sum_{\ell'\leq\ell} P(n , n -\ell') \) by \cref{lem:count-lis-permutations}. \end{proof}

%% file: conclusion.tex
In this work, we followed the very young research of jumbled scattered factors introduced in \cite{DBLP:conf/dlt/FleischmannHKK25}.
Even though the definition of jumbled scattered factors in \cite{DBLP:conf/dlt/FleischmannHKK25} is in a way a generalisation of scattered factors, we altered this definition in order to obtain $\JScatFact(w,0)=\ScatFact(w)$ for all $w\in\Sigma^{\ast}$. Moreover, we tightened the connection
between jumbled scattered factors and longest common scattered factors by introducing the language $\LCSF(w,k,\ell)$
as all words of length $k$ having a longest common scattered factor with $w$ of length least $k-\ell$.

Furthermore,
we worked in two distinct but related directions: the jumbled $k$-Simon equivalence and the number of embeddings of a jumbled scattered factor in a word. Regarding the first part, we presented three algorithms for testing the $\ell$-jumbled $k$-Simon equivalence. Specifically,
we tested the equivalence
for a small $k$ in the first algorithm,
and we tested the equivalence
for arbitrary $k$
where $k-\ell$ and $\ell$ are constant, respectively.
In the second part,
we extended the well-established research on the number of embeddings (different occurrences) from scattered factors (or even factors) to jumbled scattered factors. Therefore, we firstly introduced the notion of the jumbled binomial coefficient as a generalisation of the binomial coefficient for scattered factors. Here, we presented a recursion to compute the number of jumbled embeddings based on the well-known notion of longest increasing subsequences. In a second step, we computed the jumbled binomial coefficient based on decomposition. Lastly, we connected it to the famous Young tableaux.  
Furthermore, we contributed to the possible sizes of the jumbled Simon equivalence classes.

\medskip

\textbf{Acknowledgment.} We like to thank Melf Kammholz for sharing ideas and fruitful discussions.

\textbf{Disclaimer.} The authors ensure that the research was conducted without any AI usage.

%% file: bib.bib
@book{lothaire,
	author={M. Lothaire},
	place={Cambridge},
	series={Cambridge Mathematical Library},
	title={Combinatorics on Words},
	publisher={Cambridge University Press},
	year={1997},
	collection={Cambridge Mathematical Library}
}

@article{KarkkainenSB06,
	author    = {J. K{\"{a}}rkk{\"{a}}inen and
	P. Sanders and
	S. Burkhardt},
	title     = {Linear work suffix array construction},
	journal   = {J. {ACM}},
	volume    = {53},
	number    = {6},
	pages     = {918--936},
	year      = {2006},
}

@article{DBLP:journals/fuin/KoscheKMS22,
	author       = {M. Kosche and
	T. Ko{\ss} and
	F. Manea and
	S. Siemer},
	title        = {Absent Subsequences in Words},
	journal      = {Fundam. Informaticae},
	volume       = {189},
	number       = {3-4},
	pages        = {199--240},
	year         = {2022}
}

@article{DBLP:journals/algorithmica/BannaiIKKP24,
	author       = {H. Bannai and
	T. I and
	T. Kociumaka and
	D. K{\"{o}}ppl and
	S. J. Puglisi},
	title        = {Computing Longest {L}yndon Subsequences and Longest Common {L}yndon Subsequences},
	journal      = {Algorithmica},
	volume       = {86},
	number       = {3},
	pages        = {735--756},
	year         = {2024}
}

@article{DBLP:journals/aam/ForresterM23,
	author       = {P. J. Forrester and
	A. Mays},
	title        = {Finite size corrections relating to distributions of the length of
	longest increasing subsequences},
	journal      = {Adv. Appl. Math.},
	volume       = {145},
	pages        = {102482},
	year         = {2023}
}

@inproceedings{DBLP:conf/isaac/AdamsonFHKMN23,
	author       = {D. Adamson and
	P. Fleischmann and
	A. Huch and
	T. Ko{\ss} and
	F. Manea and
	D. Nowotka},
	title        = {k-Universality of Regular Languages},
	booktitle    = {{ISAAC}	2023},
	series       = {LIPIcs},
	volume       = {283},
	pages        = {4:1--4:21},
	publisher    = {Schloss Dagstuhl - Leibniz-Zentrum f{\"{u}}r Informatik},
	year         = {2023}
}

@article{DBLP:journals/jcss/MateescuSY04,
	author       = {A. Mateescu and
	A. Salomaa and
	S. Yu},
	title        = {Subword histories and Parikh matrices},
	journal      = {J. Comput. Syst. Sci.},
	volume       = {68},
	number       = {1},
	pages        = {1--21},
	year         = {2004}
}

@inproceedings{DBLP:conf/lics/HalfonSZ17,
	author       = {S. Halfon and
	P. Schnoebelen and
	G. Zetzsche},
	title        = {Decidability, complexity, and expressiveness of first-order logic
	over the subword ordering},
	booktitle    = {{LICS} 2017},
	pages        = {1--12},
	publisher    = {{IEEE} Computer Society},
	year         = {2017}
}

@article{DBLP:journals/ipl/KarandikarKS15,
	author       = {P. Karandikar and
	M. Kufleitner and
	P. Schnoebelen},
	title        = {On the index of {S}imon's congruence for piecewise testability},
	journal      = {Inf. Process. Lett.},
	volume       = {115},
	number       = {4},
	pages        = {515--519},
	year         = {2015}
}

@article{DBLP:journals/lmcs/KarandikarS19,
	author       = {P. Karandikar and
	P. Schnoebelen},
	title        = {The height of piecewise-testable languages and the complexity of the
	logic of subwords},
	journal      = {Log. Methods Comput. Sci.},
	volume       = {15},
	number       = {2},
	year         = {2019}
}

@InProceedings{10.1007/978-3-030-17127-8_20,
	author="Kuske, D.
	and Zetzsche, G.",
	title="Languages Ordered by the Subword Order",
	booktitle="FoSSaCS",
	year="2019",
	publisher="Springer International Publishing",
	pages="348--364"
}

@InProceedings{10.1007/978-3-030-50026-9_21,
	author="Kuske, D.",
	title="The Subtrace Order and Counting First-Order Logic",
	booktitle="{CSR}",
	year="2020",
	publisher="Springer International Publishing",
	pages="289--302"
}

@phdthesis{simon1972hierarchies,
	title={Hierarchies of events with dot-depth one},
	author={Simon, I.},
	year={1972},
	school={University of Waterloo}
}

@inproceedings{DBLP:conf/automata/Simon75,
	author       = {I. Simon},
	title        = {Piecewise testable events},
	booktitle    = {{AFL}, 2nd {GI} Conference, Kaiserslautern,
	May 20-23, 1975},
	series       = {LNCS},
	volume       = {33},
	pages        = {214--222},
	publisher    = {Springer},
	year         = {1975}
}

@InProceedings{zetzsche:LIPIcs.ICALP.2016.123,
	author =	{Zetzsche, G.},
	title =	{{The Complexity of Downward Closure Comparisons}},
	booktitle =	{ICALP 2016},
	pages =	{123:1--123:14},
	series =	{LIPIcs},
	year =	{2016},
	volume =	{55},
	publisher =	{Schloss Dagstuhl -- Leibniz-Zentrum f{\"u}r Informatik}
}

@article{DBLP:journals/tcs/Baeza-Yates91,
	author       = {R. A. Baeza{-}Yates},
	title        = {Searching Subsequences},
	journal      = {Theor. Comput. Sci.},
	volume       = {78},
	number       = {2},
	pages        = {363--376},
	year         = {1991}
}

@inproceedings{DBLP:conf/soda/BringmannK18,
	author       = {K. Bringmann and
	M. K{\"{u}}nnemann},
	title        = {Multivariate Fine-Grained Complexity of Longest Common Subsequence},
	booktitle    = {{SODA} 2018},
	pages        = {1216--1235},
	publisher    = {{SIAM}},
	year         = {2018}
}

@article{DBLP:journals/jacm/Maier78,
	author       = {D. Maier},
	title        = {The Complexity of Some Problems on Subsequences and Supersequences},
	journal      = {J. {ACM}},
	volume       = {25},
	number       = {2},
	pages        = {322--336},
	year         = {1978}
}

@article{DBLP:journals/jacm/WagnerF74,
	author       = {R. A. Wagner and
	M. J. Fischer},
	title        = {The String-to-String Correction Problem},
	journal      = {J. {ACM}},
	volume       = {21},
	number       = {1},
	pages        = {168--173},
	year         = {1974}
}

@article{DBLP:journals/bib/DoLL21,
	author       = {D. T. Do and
	T. T. Q. Le and
	N.{-}Q.{-}K. Le},
	title        = {Using deep neural networks and biological subwords to detect protein
	S-sulfenylation sites},
	journal      = {Briefings Bioinform.},
	volume       = {22},
	number       = {3},
	year         = {2021}
}

@inproceedings{DBLP:conf/aaai/WangCG20,
	author       = {C. Wang and
	K. Cho and
	J. Gu},
	title        = {Neural Machine Translation with Byte-Level Subwords},
	booktitle    = {{EAAI} 2020},
	pages        = {9154--9160},
	publisher    = {{AAAI} Press},
	year         = {2020}
}

@inproceedings{DBLP:conf/stacs/GawrychowskiKKM21,
	author       = {P. Gawrychowski and
	M. Kosche and
	T. Ko{\ss} and
	F. Manea and
	S. Siemer},
	title        = {Efficiently Testing {S}imon's Congruence},
	booktitle    = {{STACS} 2021},
	series       = {LIPIcs},
	volume       = {187},
	pages        = {34:1--34:18},
	publisher    = {Schloss Dagstuhl - Leibniz-Zentrum f{\"{u}}r Informatik},
	year         = {2021}
}

@inproceedings{DBLP:conf/mfcs/FleischerK18,
	author       = {L. Fleischer and
	M. Kufleitner},
	title        = {Testing {S}imon's congruence},
	booktitle    = {{MFCS} 2018},
	series       = {LIPIcs},
	volume       = {117},
	pages        = {62:1--62:13},
	publisher    = {Schloss Dagstuhl - Leibniz-Zentrum f{\"{u}}r Informatik},
	year         = {2018}
}

@inproceedings{DBLP:conf/dlt/KimHKS23,
	author       = {S. Kim and
	Y.{-}S. Han and
	S.{-}K. Ko and
	K. Salomaa},
	title        = {On the {S}imon's Congruence Neighborhood of Languages},
	booktitle    = {DLT},
	series       = {LNCS},
	volume       = {13911},
	pages        = {168--181},
	publisher    = {Springer},
	year         = {2023}
}

@inproceedings{DBLP:conf/isaac/DayKMS22,
  author       = {J.D. Day and
                  M. Kosche and
                  F. Manea and
                  M.L. Schmid},
  title        = {Subsequences with Gap Constraints: Complexity Bounds for Matching
                  and Analysis Problems},
  booktitle    = {{ISAAC}},
  series       = {LIPIcs},
  volume       = {248},
  pages        = {64:1--64:18},
  publisher    = {Schloss Dagstuhl - Leibniz-Zentrum f{\"{u}}r Informatik},
  year         = {2022},
}

@article{SchmidSIGMOD,
	author = {Sattler, R. and Kleest-Mei\ss{}ner, S. and Lange, S. and Schmid, M. L. and Schweikardt, N. and Weidlich, M.},
	title = {DISCES: Systematic Discovery of Event Stream Queries},
	year = {2025},
	issue_date = {February 2025},
	publisher = {Association for Computing Machinery},
	address = {New York, NY, USA},
	volume = {3},
	number = {1},
	journal = {Proc. ACM Manag. Data}
}

@inproceedings{FrochauxK23,
	author       = {A. Frochaux and
	S. Kleest{-}Mei{\ss}ner},
	title        = {Puzzling over Subsequence-Query Extensions: Disjunction and Generalised
	Gaps},
	booktitle    = {{AMW} 2023},
	series       = {{CEUR} Workshop Proceedings},
	volume       = {3409},
	publisher    = {CEUR-WS.org},
	year         = {2023},
}

@inproceedings{Kleest-Meissner23,
	author    = {S. Kleest{-}Mei{\ss}ner and
	R. Sattler and
	Markus L. Schmid and
	N. Schweikardt and
	M. Weidlich},
	title     = {Discovering Multi-Dimensional Subsequence Queries from Traces - From
	Theory to Practice},
	booktitle = {{BTW}},
	series    = {{LNI}},
	volume    = {{P-331}},
	pages     = {511--533},
	year      = {2023},
}

@book{DBLP:books/cu/Gusfield1997,
  author       = {D. Gusfield},
  title        = {Algorithms on Strings, Trees, and Sequences - Computer Science and
                  Computational Biology},
  publisher    = {Cambridge University Press},
  year         = {1997},
  isbn         = {0-521-58519-8},
}

@inproceedings{Kleest-Meissner22,
	author       = {S. Kleest{-}Mei{\ss}ner and
	R. Sattler and
	M. L. Schmid and
	N. Schweikardt and
	M. Weidlich},
	title        = {Discovering Event Queries from Traces: Laying Foundations for Subsequence-Queries
	with Wildcards and Gap-Size Constraints},
	booktitle    = {{ICDT}},
	series       = {LIPIcs},
	volume       = {220},
	pages        = {18:1--18:21},
	year         = {2022},
}

@inproceedings{artikis2017complex,
	title={Complex event recognition languages: Tutorial},
	author={Artikis, A. and Margara, A. and Ugarte, M. and Vansummeren, S. and Weidlich, M.},
	booktitle={DEBS},
	pages={7--10},
	year={2017}
}

@article{DBLP:journals/lmcs/BaumannGTZ23,
	author       = {P. Baumann and
	M. Ganardi and
	R. S. Thinniyam and
	G. Zetzsche},
	title        = {Existential Definability over the Subword Ordering},
	journal      = {Log. Methods Comput. Sci.},
	volume       = {19},
	number       = {4},
	year         = {2023}
}

@inproceedings{DBLP:conf/dlt/FleischmannHKK25,
  author       = {P. Fleischmann and
                  A. Huch and
                  M. Kammholz and
                  T. Ko{\ss}},
  editor       = {Sang{-}Ki Ko and
                  Florin Manea},
  title        = {Jumbled Scattered Factors},
  booktitle    = {{DLT}},
  series       = {Lecture Notes in Computer Science},
  volume       = {16036},
  pages        = {261--277},
  publisher    = {Springer},
  year         = {2025},
}

@inproceedings{DBLP:conf/cwords/Puzynina19,
	author       = {S. Puzynina},
	title        = {Abelian Properties of Words},
	booktitle    = {{WORDS} 2019},
	series       = {LNCS},
	volume       = {11682},
	pages        = {28--45},
	publisher    = {Springer},
	year         = {2019}
}

@article{DBLP:journals/tcs/KarhumakiPRW17,
	author       = {J. Karhum{\"{a}}ki and
	S. Puzynina and
	M. Rao and
	M. A. Whiteland},
	title        = {On cardinalities of k-abelian equivalence classes},
	journal      = {Theor. Comput. Sci.},
	volume       = {658},
	pages        = {190--204},
	year         = {2017}
}

@article{higman1952ordering,
	title={Ordering by divisibility in abstract algebras},
	author={Higman, G.},
	journal={Proceedings of the London Mathematical Society},
	volume={3},
	number={1},
	pages={326--336},
	year={1952},
	publisher={Wiley Online Library}
}

@article{haines1969free,
	title={On free monoids partially ordered by embedding},
	author={Haines, L.H.},
	journal={J. Comb. Theory. Ser.},
	volume={6},
	number={1},
	pages={94--98},
	year={1969},
	publisher={Elsevier}
}

@article{Schensted_1961, title={Longest Increasing and Decreasing Subsequences}, volume={13}, journal={Canadian Journal of Mathematics}, author={Schensted, C.}, year={1961}, pages={179–191}}

@article{FRT54, title={The Hook Graphs of the Symmetric Group}, volume={6}, DOI={10.4153/CJM-1954-030-1}, journal={Canadian Journal of Mathematics}, author={Frame, J. S. and Robinson, G. de B. and Thrall, R. M.}, year={1954}, pages={316–324}}

@inproceedings{SchwentickTV01,
  author       = {Thomas Schwentick and
                  Denis Th{\'{e}}rien and
                  Heribert Vollmer},
  title        = {Partially-Ordered Two-Way Automata: {A} New Characterization of {DA}},
  booktitle    = {Proceedings of the 5th International Conference on Developments in Language Theory, {DLT}},
  pages        = {239--250},
  year         = {2001}
}

@article{DBLP:journals/ijfcs/HolzerK09,
  author       = {M. Holzer and
                  M. Kutrib},
  title        = {Nondeterministic Finite Automata - Recent Results on the Descriptional
                  and Computational Complexity},
  journal      = {Int. J. Found. Comput. Sci.},
  volume       = {20},
  number       = {4},
  pages        = {563--580},
  year         = {2009}
}

@article{DBLP:journals/ipl/Birget92,
  author       = {J.{-}C. Birget},
  title        = {Intersection and Union of Regular Languages and State Complexity},
  journal      = {Inf. Process. Lett.},
  volume       = {43},
  number       = {4},
  pages        = {185--190},
  year         = {1992}
}


%% file: shortcortana.bib
@string{may	= {Mai}}

@string{may	= {}}

@string{LNCS = {LNCS}}

@string{LIPIcs = {LIPIcs}}

@string{STACS = {Symposium on Theoretical Aspects of Computing}}

@string{ICALP = {International Colloquium on Automata, Languages, and Programming}}

@string{MFCS = {Mathematical Foundations of Computer Science}}

@string{Springer = {Springer}}

@string{PPL  = {Par. Process. Lett.}}
